\documentclass[11pt, letterpaper]{article}
\usepackage{times}
\usepackage{amsmath,amsfonts,amssymb,amsthm}
\usepackage[colorlinks=true,linkcolor=red,citecolor=blue,urlcolor=red]{hyperref}
\usepackage{paralist}
\usepackage{framed}
\usepackage{verbatim}

\usepackage{mathptmx}
\usepackage[text={6.4in,8.9in}]{geometry}

\newenvironment{reminder}[1]{\medskip
\noindent {\bf Reminder of #1  }\em}{\medskip}

\newtheorem{theorem}{Theorem}[section]

\newtheorem{corollary}{Corollary}[section]

\newtheorem{definition}{Definition}[section]

\def \MCE {\text{\sc Multipoint Circuit Evaluation}}
\def \size {\text{size}}

\def \NP {{\sf NP}}

\def \NTIME {{\sf NTIME}}

\def\eps{\varepsilon}

\def \PSPACE {{\sf PSPACE}}
\def\poly{\text{poly}}
\def \mult {\text{mult}}

\def \E {\text{\sf E}}
\def \P {\text{\sf P}}

\def \Z {{\mathbb Z}}
\def \R {{\mathbb R}}
\def \F {{\mathbb F}}

\def \MA {{\sf MA}}

\title{Strong ETH Breaks With Merlin and Arthur: \\
 Short Non-Interactive Proofs of Batch Evaluation}

\author{
Ryan Williams\footnote{Computer Science Department, Stanford University. Part of this work was done while visiting the Simons Institute for the Theory of Computing, Berkeley, CA. Supported in part by a David Morgenthaler II Faculty Fellowship, a Sloan Fellowship, and NSF CCF-1212372. Any opinions, findings and conclusions or recommendations expressed in this material are those of the authors and do not necessarily reflect the views of the National Science Foundation.} 
}

\begin{document}
\date{}
\maketitle

\begin{abstract} We present an efficient proof system for {\sc Multipoint Arithmetic Circuit Evaluation}: for any arithmetic circuit $C(x_1,\ldots,x_n)$ of size $s$ and degree $d$ over a field $\F$, and any inputs $a_1,\ldots,a_K \in \F^n$,
\begin{itemize}
\item the Prover sends the Verifier the values $C(a_1), \ldots, C(a_K) \in \F$ and a proof of $\tilde{O}(K \cdot d)$ length, and 
\item the Verifier tosses $\poly(\log(dK|\F|/\eps))$ coins and can check the proof in about $\tilde{O}(K \cdot(n + d) + s)$ time, with probability of error less than $\eps$.
\end{itemize}
For small degree $d$, this ``Merlin-Arthur'' proof system (a.k.a. MA-proof system) runs in nearly-linear time, and has many applications. For example, we obtain MA-proof systems that run in $c^{n}$ time (for various $c < 2$) for the Permanent, $\#$Circuit-SAT for all sublinear-depth circuits, counting Hamiltonian cycles, and infeasibility of $0$-$1$ linear programs. In general, the value of any polynomial in Valiant's class ${\sf VP}$ can be certified faster than ``exhaustive summation'' over all possible assignments. These results strongly refute a Merlin-Arthur Strong ETH and Arthur-Merlin Strong ETH posed by Russell Impagliazzo and others. 

We also give a three-round (AMA) proof system for quantified Boolean formulas running in $2^{2n/3+o(n)}$ time, 
nearly-linear time MA-proof systems for counting orthogonal vectors in a collection and finding Closest Pairs in the Hamming metric, and a MA-proof system running in $n^{k/2+O(1)}$-time for counting $k$-cliques in graphs. 

We point to some potential future directions for refuting the Nondeterministic Strong ETH.

\end{abstract}

\section{Introduction}

Suppose you have a circuit of size $s$ that you want to evaluate on $k$ different inputs. In the worst case, you'd expect and need $O(s \cdot k)$ time to do this yourself. What if you asked a powerful computer to evaluate the circuit for you? The computer may be extremely fast relative to you, and send you the $k$ answers almost immediately. But how can you (quickly) check that the computer used \emph{your} circuit, and didn't just make up the answers? Such ``delegating/verifiable computation'' questions naturally arise in the study of interactive proofs, and have recently seen increased attention in the crypto community
(see~\cite{GoldwasserKR15,GennaroGP10,ChungKV10,ApplebaumIK10,FioreG12,
Thaler13,
KalaiRR14}~for a sample of the different models and goals). 

For circuits with a certain natural structure\footnote{In particular, the proof system works for all arithmetic circuits using addition and multiplication over a finite field, where the resulting polynomial has low degree. A surprising number of functions can be efficiently implemented in this way.}, we show in this paper how a powerful computer can \emph{very} efficiently prove in one shot (with extremely low probability of error) that its answers are indeed the outputs of your circuit. Omitting low-order terms, the proof is about $\tilde{O}(s+k)$ bits long, and takes about $\tilde{O}(s+k)$ time to verify---roughly proportional to the size of the circuit and the $k$ inputs. The proof system is simple and has no nasty hidden constants, low randomness requirements, and many theoretical applications. 

\subsection{Our Results}

Our evaluation result is best phrased in terms of arithmetic circuits over plus and times gates, evaluated over a finite field. We consider the problem of evaluating such a circuit on many inputs in batch:

\begin{definition} The $\MCE$ problem: given an arithmetic circuit $C$ on $n$ variables over a finite field $\F$, and a list of inputs $a_1,\ldots,a_K \in \F^n$, output $\left(C(a_1),\ldots,C(a_K)\right)\in \F^K$.
\end{definition}

An important special case of $\MCE$ is when the arithmetic circuit is a sum of products of variables (a $\Sigma \Pi$ circuit). This version is called {\sc Multivariate Multipoint Evaluation} by Kedlaya and Umans~\cite{KedlayaU11}; they give the best known algorithms for this case, showing how to solve it in about $(d^n+K)^{1+o(1)}\poly(\log m)$ time over $\Z_m$, where $d$ is the degree of each variable and $n$ is the number of variables. 
The simplest instance of multipoint evaluation considers circuits that are a sum of products of \emph{one} variable; this case is well-known to have very efficient algorithms (see Section~\ref{poly-prelims}). However, for more expressive circuits (such as $\Sigma \Pi \Sigma$, sums of products of sums), no significant improvements over the obvious batch evaluation algorithm have been reported.

Our first result is that multipoint evaluation of \emph{general} arithmetic circuits of low degree can be ``delegated'' very efficiently, in a publicly verifiable and non-interactive way:

\begin{theorem} \label{main-multipoint} For every finite field $\F$ and $\eps > 0$, $\MCE$ for $K$ points in $\F^n$ on a circuits of $n$ inputs, $s$ gates, and degree $d$ has an probabilistic verifier $V$ where, for every circuit $C$,
\begin{compactitem}
\item There is a unique proof of $(C(a_1),\ldots,C(a_K))$ that is  $\tilde{O}(K \cdot d)$ bits long\footnote{The $\tilde{O}$ omits polylog factors in $K$, $|\F|$, $d$, $s$, and $1/\eps$.}, and  
\item The proof can be verified by $V$ with access to $C$, $\tilde{O}(1)$ bits of randomness, and $\tilde{O}(K\cdot \max\{d,n\} + s)$ time, such that $(C(a_1),\ldots,C(a_K))$ is output incorrectly with probability at most $\eps$.
\end{compactitem}
\end{theorem}

The proof system is fairly simple to motivate. We want the proof to be a succinct representation of the circuit $C$ that is both easy to evaluate on all of the $K$ given inputs, and also easy to verify with randomness. We will set the proof to be a univariate polynomial $Q(x)$ defined over a sufficiently large extension field of $\F$, of degree about $K \cdot d$, that ``sketches'' the evaluation of the degree-$d$ arithmetic circuit $C$ over all $K$ assignments. The polynomial $Q$ satisfies two conflicting conditions:
\begin{compactenum}
\item The verifier can use the sketch $Q$ to efficiently produce the truth table of $C$. In particular, for some explicitly chosen $\alpha_i$ from the extension of $\F$, we have $(Q(\alpha_0), Q(\alpha_1),\ldots,Q(\alpha_K)) = (C(a_1),\ldots,C(a_K))$. 
\item The verifier can check that $Q$ is a faithful representation of $C$'s behavior on the list of $K$ inputs in about $K + |C|$ time, with randomness.
\end{compactenum}
The construction of $Q$ uses an trick originating from the holographic proofs of Babai {\em et al.}~\cite{Babai-Fortnow-Levin-Szegedy91}, in which multivariate expressions are efficiently ``expressed'' as univariate ones. Both of the two items utilize fast algorithms for manipulating univariate polynomials. In the parlance of interactive proofs, Theorem~\ref{main-multipoint} gives a \emph{Merlin-Arthur proof system} for batch evaluation (Merlin is the prover, Arthur is the verifier, and Merlin communicates first).

{\bf Applications to Some Exponential Time Hypotheses.} The results of this paper were originally motivated by attempts to refute \emph{exponential time hypotheses} of increasing strength. The Exponential Time Hypothesis (ETH)~\cite{IPZ01} is that 3-SAT requires $2^{\eps n}$ time for some $\eps > 0$; ETH has been singularly influential in the area of exact algorithms for $\NP$-hard problems (see~\cite{LokshtanovMS11survey} for a survey). A more fine-grained version of ETH is the Strong Exponential Time Hypothesis (SETH)~\cite{IP01,CIP09}, which further asserts that $k$-SAT  requires $2^{n-o(n)}$ time for unbounded $k$. SETH has also been a powerful driver of research in the past several years, especially with its connections to the solvability of basic problems in $\P$ (see the recent survey~\cite{VVWilliams15}).

Recently, Carmosino \emph{et al.}~\cite{CarmosinoGIMPS15} proposed the Nondeterministic Strong ETH (NSETH): refuting unsatisfiable $k$-CNFs requires nondeterministic $2^{n-o(n)}$ time for unbounded $k$. Put another way, NSETH says there are no proof systems that can refute unsatisfiable $k$-SAT instances significantly more efficiently than enumeration of all variable assignments. The NSETH is quite consistent with known results in proof complexity~\cite{PudlakI00,BeckI13}. Earlier, Carmosino \emph{et al.} (private communication) also proposed a Merlin-Arthur and Arthur-Merlin Strong ETH (MASETH and AMSETH, respectively) which assert that no $O(1)$-round probabilistic proof systems can refute unsatisfiable $k$-CNFs in $2^{n-\Omega(n)}$ time.

Our first application of Theorem~\ref{main-multipoint} is a strong refutation of MASETH and AMSETH:

\begin{theorem}[MASETH is False] \label{MASETH} There is a probabilistic verifier $V$ where, for every Boolean circuit $C$ on $n$ variables of $o(n)$ depth and bounded fan-in,
\begin{compactitem}
\item There is an $O^{\star}(2^{n/2})$-bit proof that the number of SAT assignments to $C$ is a claimed value\footnote{The $O^{\star}$ notation omits polynomial factors in $n$.}, and  
\item The proof can checked by $V$ with access to $C$, using $O(n)$ bits of randomness and $O^{\star}(2^{n/2})$ time, with probability of error at most $1/\poly(n)$.
\end{compactitem}
\end{theorem}

That is, one can refute UNSAT circuits of $2^{o(n)}$ size and $o(n)$ depth significantly faster than brute force enumeration, using a small amount of randomness in verification. Analogues of Theorem~\ref{MASETH} hold for other $\#P$-complete problems: for instance, the Permanent can be certified in $O^{\star}(2^{n/2})$ time, and the number of Boolean feasible solutions to a linear program can be certified in $O^{\star}(2^{3n/4})$.
In fact, if we allow the proof to depend on $O(n)$ coins tossed prior to sending the proof, one can also solve \emph{Quantified Boolean Formulas} (QBF) faster:

\begin{theorem} \label{QBF}
QBFs with $n$ variables and $m \leq 2^n$ connectives have a three-round $2^{2n/3} \cdot \poly(n,m)$ time interactive proof system using $O(n)$ bits of randomness.
\end{theorem}

A seminal result in interactive computation is that $\PSPACE={\sf IP}$; that is, polynomial space captures interactive proof systems that use $\poly(n)$ time and $\poly(n)$ rounds~\cite{Shamir90}. Theorem~\ref{QBF} shows how three rounds of interaction can already significantly reduce the cost of evaluating $\PSPACE$-complete problems. From these results, we see that either $O(n)$ bits of randomness can make a substantial difference in the proof lengths of $n$-bit propositions, or the Nondeterministic SETH is false. In fact, one can isolate a simple \emph{univariate polynomial identity testing} problem that is solvable in $\tilde{O}(n)$ randomized time and $\tilde{O}(n^2)$ time deterministically, but an $n^{1.999}$-time nondeterministic algorithm would refute NSETH; see Section~\ref{upit-nseth}.

\paragraph{Applications to Some Polynomial-Time Problems.} In Appendix~\ref{ov-section}, we apply Theorem~\ref{main-multipoint} to a group of problems at the basis of a recent theory of ``hardness within $\P$''~\cite{VVWilliams15}. A central problem in this theory is {\sc Orthogonal Vectors}, which asks if there is an orthogonal pair among $n$ Boolean vectors in $d$ dimensions~\cite{Williams05,
RoddityV13,Williams-Yu14,Bringmann14,AbboudV14,AbboudWY15,BackursI15,AbboudBV15}. The \emph{OV conjecture} is that this problem cannot be solved in $n^{2-\eps}\cdot 2^{o(d)}$, for every $\eps > 0$. It is known that SETH implies the OV conjecture~\cite{Williams05,Williams-Yu14}. The OV conjecture can also be refuted in the Merlin-Arthur setting, in the following sense:

\begin{theorem}\label{ov} Let $d \leq n$. There is an MA-proof system such that for every $A \subseteq \{0,1\}^d$ with $|A|=n$, the verifier certifies the number of orthogonal pairs in $A$, running in $\tilde{O}(n\cdot d)$ time with error probability $1/\poly(n)$.
\end{theorem}

Because several basic problems in $\P$ can be subquadratic-time reduced to {\sc Orthogonal Vectors} (see the above references and Appendix~\ref{ov-section}), Theorem~\ref{ov} implies subquadratic-time MA-proof systems for these problems as well. To give another example, we also obtain a nearly-linear time proof system for verifying {\sc Closest Pairs} in the Hamming metric: 

\begin{theorem}\label{hamming} Let $d \leq n$. There is an MA-proof system such that for every $A \subseteq \{0,1\}^d$ with $|A|=n$, and every given parameter $k \in \{0,1,\ldots,d\}$, the verifier certifies for all $v \in A$ the number of points $w \in A$ with Hamming distance at most $k$ from $v$, running in $\tilde{O}(n\cdot d)$ time with error probability $1/\poly(n)$.
\end{theorem}

The best known randomized algorithm for Hamming nearest neighbors only runs in $o(n^2)$ time when $d = o(\log^2 n/\log \log n)$~\cite{AlmanW15}. Finally, we also give an efficient proof system for the $k$-clique problem:

\begin{theorem}\label{k-clique} For every $k$, there is a MA-proof system such that for every graph $G$ on $n$ nodes, the verifier certifies the number of $k$-cliques in $G$ using $\tilde{O}(n^{\lfloor k/2\rfloor+2})$ time, with error probability $1/\poly(n)$.
\end{theorem}

\section{Preliminaries}

For a vector $v \in D^d$ for some domain $D$, we let $v[i] \in D$ denote the $i$th component of $v$. We assume basic familiarity with Computational Complexity, especially the theory of interactive proofs and Merlin-Arthur games as initiated by Goldwasser-Micali-Rackoff~\cite{GoldwasserMR85} and Babai~\cite{Babai85} (see Arora and Barak~\cite{Arora-Barak09}, Chapter 8). All of the interactive proofs (also known as ``protocols'') of this paper will use \emph{public} randomness, visible to the Prover (also known as ``Merlin'') and the Verifier (also known as ``Arthur''). Along the way, we will recall some particulars of known results as needed. 

\label{poly-prelims}
\paragraph{Some Algorithms for Polynomial Computations.} We need some classical results in algebraic complexity (see also von zur Gathen and Gerhard~\cite{vzG-G13}). Let $\F$ be an arbitrary field, and let $\mult(n) = O(n \log^2 n)$ be the time needed to multiply two degree-$n$ univariate polynomials. 

\begin{theorem}[Fast Multipoint Evaluation of Univariate Polynomials~\cite{Fiduccia72}] \label{multipoint-univariate} Given a polynomial $p(x) \in \F[X]$ with $\deg(p) \leq n$, presented as a vector of coefficients $[a_0,\ldots,a_{\deg(p)}]$, and given points $\alpha_1,\ldots,\alpha_n \in \F$, we can output the vector $(p(\alpha_1),\ldots,p(\alpha_n)) \in \F^n$ in $O(\mult(n) \cdot \log n)$ additions and multiplications in $\F$. 
\end{theorem}

\begin{theorem}[Fast Univariate Interpolation~\cite{Horowitz72}] \label{interpolation} Given a set of pairs $\{(\alpha_1,\beta_1),\ldots,(\alpha_n,\beta_n)\}$ with all $\alpha_i$ distinct, we can output the coefficients of $p(x) \in \F[X]$ of degree at most $n$ satisfying $p(\alpha_i) = \beta_i$ for all $i$, in $O(\mult(n) \cdot \log n)$ additions and multiplications in $\F$. 
\end{theorem}

\subsection{More Related Work}

Besides what we have already mentioned, there is a vast body of work on non-interactive probabilistic protocols and delegating computation which we are ill-equipped to cover in detail. We confine ourselves to discussing results that seem closest to the present work.\footnote{We would be happy to hear of results related to ours that we did not cite.} 

There has been much work on bounding the communication between the prover and verifier. For instance, this is not the first time that Merlin and Arthur have led to an unexpected square-root speedup: Aaronson and Wigderson~\cite{Aaronson-Wigderson09} gave an MA communication protocol for computing the inner product of two $n$-length vectors which runs in $\tilde{O}(\sqrt{n})$ time. Their protocol uses a nice bivariate encoding of vectors, although it is somewhat different from ours (which is univariate). Gur and Rothblum~\cite{GurR15} obtain a similar square-root speedup for checking sums in the ``non-interactive property testing'' setting. Goldreich and Hastad~\cite{GoldreichH98} and Goldreich, Vadhan, and Wigderson~\cite{GoldreichVW02} studied interactive proofs which seek to minimize the number of bits sent from Merlin to Arthur. The ``small bits'' case is of course even more restrictive than the ``small rounds'' case. The latter reference shows that for any language $L$ that has an interactive proof with $b$ bits of communication, there is an $O(1)$-round interactive proof for $\overline{L}$ that uses only $\exp(b)$ communication. The authors also conjectured an ``Arthur-Merlin ETH'' that $\#$SAT does not have a $2^{o(n)}$-time AM-proof system with $O(1)$ rounds. What we report in this paper is rather far from disproving this ``AMETH'' conjecture, but it is interesting that some non-trivial progress can be made.

Goldwasser, Kalai, and Rothblum~\cite{GoldwasserKR15} study what they call \emph{delegating computation}, proving (for example) that for all logspace-uniform NC circuits $C$, one can prove that $C(x) = 1$ on an input $x$ of length $n$ with $\tilde{O}(n)$ verification time, $O(\log n)$ space, and $\poly(\log n)$ communication complexity between the prover and verifier. Despite the amazingly low running time and space usage, the protocols of this work are highly non-interactive: they need $\poly(\log n)$ \emph{rounds} between the prover and verifier as well.

Relating our work to proof complexity, Grochow and Pitassi~\cite{GrochowP14} introduced a new algebraic proof system based on axioms satisfied by any Boolean circuit that solves the polynomial identity testing problem. The proofs in their system can be efficiently verified by running a polynomial identity test, implying they can be viewed as proof of a Merlin-Arthur type. An intriguing property of their proof system is that super-polynomial lower bounds for it would prove lower bounds for the Permanent. 
    
The area of \emph{verifiable computation} (e.g.~\cite{Parno2013pinocchio}) is a new subject in cryptography, and is certainly related to our work. However, in crypto the work appears to be either very specific to particular functions, or it relies on very heavy machinery like probabilistically checkable proofs, or it relies on cryptographic hardness assumptions. 

In our setting, we want non-interactive proofs for batch computations that are \emph{shorter} than the computation time, with the typical ``perfect completeness'' and ``low error soundness'' conditions preserved, and which work unconditionally.

\section{Fast Multipoint Circuit Evaluation (With Merlin and Arthur)}

In this section, we give the proof system for multipoint arithmetic circuit evaluation:

\begin{theorem} \label{multipoint-eval} For every prime power $q$ and $\eps > 0$, $\MCE$ for $K$ points in $(\F_q)^n$ on an arithmetic circuit $C$ of $n$ inputs, $s$ gates, and degree $d$ has an MA-proof system where:
\begin{compactitem}
\item Merlin sends a proof of $O(K \cdot d \cdot \log(Kqd/\eps))$ bits, and  
\item Arthur tosses at most $\log(Kqd/\eps)$ coins, outputs $(C(a_1),\ldots,C(a_K))$ incorrectly with probability at most $\eps$, and runs in time $(K\cdot \max\{d,n\} + s \cdot \poly(\log s))\cdot \poly(\log (Kqd/\eps))$. 
\end{compactitem}
\end{theorem}

We have stated the theorem at this level of generality because we need good bounds on the parameters to obtain certain consequences. For example, in our proof system for quantified Boolean formulas (Theorem~\ref{QBF}), the parameters $s$, $K$, $q$, and $d$ are all various exponentials in $n$.

Because instances of $\MCE$ have length $O((K\cdot n + s \log s) \cdot \log q)$, the running time of Theorem~\ref{multipoint-eval} is essentially linear in the input length, up to the factor of $d$ in Merlin's proof (in general, $d$ could be much larger than $n$). So Theorem~\ref{multipoint-eval} is extremely powerful for arithmetic circuits of low degree.

\begin{proof} Let $q$ be a prime power and $C$ be an arithmetic circuit over $\F_q$ with degree $d$, $s$ gates, and $n$ variables. Let $a_1,\ldots,a_K \in \F_q^n$; we want to know $C(a_1),\ldots,C(a_K) \in \F_q$.

Let $\eps > 0$ be arbitrarily small, and let ${\ell}$ be the smallest integer such that $q^{\ell} > (d \cdot K)/\eps$. Let $F$ be the extension field $\F_{q^{\ell}}$. Note we can construct $\F_{q^{\ell}}$ rather quickly in the following way: Merlin can send an irreducible polynomial $f(x) \in \F_q[x]$ of degree $\ell$, and irreducibility of $f$ can be checked by running Kedlaya-Umans' deterministic irreducibility test in $\ell^{1+o(1)}\log^{2+o(1)} q$ time~(\cite{KedlayaU11}, Section 8.2). 

Since $q^{\ell} \leq (q \cdot K \cdot d)/\eps$, addition and multiplication in $F$ can be done in $(\log |F|)^{1+o(1)} \leq \log (Kqd/\eps)^{1+o(1)}$ time. Let $S \subseteq F$ be an arbitrary subset of cardinality $K$. For all $i=1,\ldots,K$, associate each vector $a_i \in (\F_q)^n$ with a unique element $\alpha_{i} \in S$, and inversely associate each $\alpha \in S$ with a unique vector $a_{\alpha} \in (\F_q)^n$. This mapping and its inverse can be easily constructed by listing the first $K$ elements of $F$ under some canonical ordering.

For all $j=1,\ldots,n$, we define $\Psi_j : F \rightarrow F$ as functions satisfying $\Psi_j(\alpha) = a_{\alpha}[j]$ for every $\alpha \in S$. That is, $\Psi_j(s)$ outputs the $j$th component of the vector $a_{\alpha} \in \F_q^n$ associated with $\alpha \in S$. Since each $\Psi_j$ is defined by $K$ input/output pairs, the $\Psi_j$ can be instantiated as polynomials of degree at most $K$. By efficient polynomial interpolation (Theorem~\ref{interpolation}), the degree-$K$ polynomials $\Psi_j(x) \in F[x]$ for all $j=1,\ldots,n$ can be constructed in $n \cdot K \cdot \poly(\log K)$ additions and multiplications.

Define the univariate polynomial $R(x) := C(\Psi_1(x),\ldots,\Psi_{n}(x))$ over $F$. By the construction of $\Psi_j$, we see that for all $i=1,\ldots,K$, $R(\alpha_{i}) = C(\Psi_1(\alpha_{i}),\ldots,\Psi_{n}(\alpha_{i})) = C(a_i[1],\ldots,a_i[n]) = C(a_i)$.  Furthermore, $\deg(R) \leq \deg(C) \cdot \left(\max_j \Psi_j\right) \leq d \cdot K$. 

Now we describe the protocol.

\begin{compactenum}
\item Merlin sends the coefficients of a polynomial $Q(x)$ over $F$ of degree at most $d \cdot K$, encoded in $d \cdot K \cdot \log(|F|)$ bits. Merlin claims that $Q(x) = R(x)$, as defined above. 

\item Arthur picks a uniform random $r \in F$ (taking at most $\log (Kqd/\eps)$ bits to describe), and wishes to check that \[Q(r) = R(x) := C(\Psi_1(r),\ldots,\Psi_{n}(r)),\] over $F$. Evaluating $Q(r)$ takes $d \cdot K \cdot (\log (Kqd/\eps))^{1+o(1)}$ time, by Horner's method. We claim that $R(r)$ can be computed in $(K\cdot n + s) \cdot (\log |F|)^{1+o(1)}$ time. First, the $n$ polynomials $\Psi_j$ of degree $K$ can be constructed in $n \cdot K \cdot \poly(\log K)$ additions and multiplications (as described above). Given the coefficients of the $\Psi_j$ polynomials, computing all values $v_j := \Psi_j(r)$ can be done straightforwardly in $O(K \cdot n)$ additions and multiplications, by producing the powers $r^0, r^1,\ldots,r^{K}$ and then computing $n$ linear combinations of these powers. (Note that each resulting value $v_j$ takes $O(\log |F|) \leq \poly(Kqd/\eps)$ bits to represent.) Then Arthur computes $C(v_1,\ldots,v_n)$ in $s \cdot \poly(\log s)$ additions and multiplications, by simple circuit evaluation over $F$. The total running time is $(K\cdot n + s \cdot \poly(\log s))\cdot (\log |F|)^{1+o(1)}$.

\item Arthur \emph{rejects the proof} if $Q(r) \neq C(v_1,\ldots,v_n)$; otherwise, he uses univariate multipoint evaluation (Theorem~\ref{multipoint-univariate}) to compute $(Q(\alpha_1),\ldots,Q(\alpha_K))$, in $K\cdot d \cdot \poly(\log (Kd)) \cdot (\log |F|)^{1+o(1)}$ time.
\end{compactenum}

On the one hand, if Merlin sends $Q(x) := R(x)$, then Arthur always outputs the tuple \[(R(\alpha_1),\ldots,R(\alpha_K)) = (C(a_1),\ldots,C(a_n)),\] regardless of the $r \in F$ chosen. On the other, if Merlin sends a ``bad'' polynomial $Q(x) \neq R(x)$ and Arthur fails to pick an $r \in F$ such that $Q(r) \neq R(r)$, then Merlin may convince Arthur of an incorrect $K$-tuple $(Q(\alpha_1),\ldots,Q(\alpha_K))$. However, since the degrees of $Q$ and $R$ are both at most $d \cdot K$, this failure of Arthur occurs with probability at most $(d \cdot K)/q^{\ell} < \eps$. \end{proof}

\subsection{Evaluating Sums Over Polynomials}

The multipoint evaluation protocol of Theorem~\ref{multipoint-eval} can be applied to perform a one-round ``sum-check'' faster than the obvious algorithm:

\begin{theorem}\label{VP-sim} Given a prime $p$, an $\eps > 0$, and an arithmetic circuit $C$ with degree $d$, $s \geq n$ gates, and $n$ variables, the sum \[\sum_{(b_1,\ldots,b_n) \in \{0,1\}^n} C(b_1,\ldots,b_n) \bmod p\] can be computed by a Merlin-Arthur protocol running in $2^{n/2} \cdot \poly(n,s,d,\log(p/\eps))$ time tossing only $n/2+O(\log(pd/\eps))$ coins, with probability of error $\eps$.
\end{theorem}

Therefore, every polynomial in the class ${\sf VP}$ (\cite{Valiant79a,Valiant82}) has a MA-proof system that beats exhaustive search in a strong sense. 

\begin{proof} (of Theorem~\ref{VP-sim}) For simplicity, assume $n$ is even. Given an arithmetic circuit $C$ for which we wish to evaluate its sum over all Boolean inputs, define the $n/2$-variable circuit \[C'(x_1,\ldots,x_{n/2}) := \sum_{(b_1,\ldots,b_{n/2}) \in \{0,1\}^{n/2}} C(x_1,\ldots,x_{n/2},b_1,\ldots,b_{n/2}).\] Note that $\deg(C') = d$ and $\size(C') \leq 2^{n/2}\cdot s$. In order to compute the full sum of $C(b_1,\ldots,b_n)$ over all $2^n$ Boolean points, it suffices to evaluate $C'$ on all of its $K := 2^{n/2}$ Boolean points $a_1,\ldots,a_{K} \in \{0,1\}^n$. 

Applying the batch evaluation protocol of Theorem~\ref{multipoint-eval}, there is an MA-proof system where Merlin sends a proof of $2^{n/2} \cdot d \cdot \poly(n,\log (pd/\eps))$ bits, then Arthur tosses $n/2 + \log (pd/\eps)$ coins, runs in $(2^{n/2}\cdot\max\{n,d\} + 2^{n/2} \cdot s \cdot \poly(\log s))\cdot \poly(n,\log (pd/\eps))$ time, and outputs $(C'(a_1),\ldots,C(a_{2^{n/2}}))$ incorrectly with probability at most $\eps$. The result follows.
\end{proof}

Two important corollaries of Theorem~\ref{VP-sim} are $O^{\star}(2^{n/2})$-time proof systems for the Permanent and $\#$SAT problems. The result for Permanent follows immediately from Ryser's formula~\cite{Ryser63}, which shows that the permanent of any $n\times n$ matrix $M$ can be written in the form \[\sum_{(a_1,\ldots,a_n) \in \{0,1\}^n} C_M(a_1,\ldots,a_n),\] where $C_M$ is a $\poly(n)$-size arithmetic circuit of degree $O(n)$ that can be determined from $M$ in $\poly(n)$ time. We describe the $\#SAT$ protocol in detail:

\begin{theorem} \label{counting-SAT} For any $k > 0$, $\# SAT$ for Boolean formulas with $n$ variables and $m$ connectives has an MA-proof system using $2^{n/2}\cdot \poly(n,m)$ time with randomness $O(n)$ and error probability $1/\exp(n)$.
\end{theorem}

\begin{proof} Let $F$ be a Boolean formula over AND, OR, and NOT with $n$ variables and $m$ connectives. First, any Boolean formula $F$ can be ``re-balanced'' as in the classical results of Brent~\cite{Brent74} and Spira~\cite{Spira71}, obtaining in $\poly(m)$  time a formula $F'$ equivalent to $F$, where $F'$ has depth at most $c\log m$ and at most $m^c$ connectives for some constant $c > 0$.

Next, we replace each AND, OR, and NOT gate of $F'$ with an equivalent polynomial of degree $2$, by the usual ``arithmetization.'' More precisely, each $OR(x,y)$ is replaced with $x + y - x \cdot y$, each $AND(x,y)$ is replaced with $x \cdot y$, and each $NOT(1-x)$ is replaced with $1-x$. The resulting arithmetic formula $P(x_1,\ldots,x_n)$ computes $F'(b_1,\ldots,b_n) = P(b_1,\ldots,b_n)$ for every $(b_1,\ldots,b_n) \in \{0,1\}^n$. Furthermore, due to the re-balancing step and the fact that every gate has outdegree $1$, we have $\deg(P) \leq 2^{c\log m} \leq m^{O(1)}$ (note the worst case is when every gate is an AND). 

Set $p > 2^n$ to be prime; note by Bertrand's postulate we may assume $p < 2^{n+1}$. We can always find such a prime deterministically in $2^{n/2+o(n)}$ time by an algorithm of Lagarias and Odlyzko~\cite{LagariasO87}. (Alternatively, the prover could send $p$ to the verifier, along with a deterministically verifiable $\poly(n)$-length proof of primality~\cite{Pratt75}.) Then $F'$ has exactly $r$ satisfying assignments if and only if \[\sum_{(b_1,\ldots,b_n) \in \{0,1\}^n} P(b_1,\ldots,b_n) = r \bmod p.\] Since $\deg(P) \leq m^{O(1)}$, we can apply Theorem~\ref{VP-sim} directly and obtain the result.
\end{proof}

Another corollary of Theorem~\ref{VP-sim} is that Merlin and Arthur can also count Hamiltonian cycles in $n$-node graphs in $O^{\star}(2^{n/2})$ time, by construing the inclusion-exclusion method of Karp~\cite{Karp82} running in $O^{\star}(2^n)$ time as a sum over $2^n$ Boolean values on an arithmetic circuit of $\poly(n)$ size. In particular, Karp's algorithm works by counting the $n$-step walks in a graph, then subtracting the count of $n$-step walks that miss at least one node, adding back the count of $n$-step walks that miss at least two nodes, etc. Each of these counts is computable by a single arithmetic circuit $C(y_1,\ldots,y_n)$ of $O(n^4)$ size which, on the input $y \in \{0,1\}^n$, counts the $n$-step walks over the subgraph of $G$ defined by the vector $y$ (negating the count if $y$ has an odd number of zeroes).

Theorem~\ref{counting-SAT} shows that Merlin and Arthur can count the number of satisfying assignments to Boolean formulas of $2^{\delta n}$ size in $2^{n(1/2+O(\delta))}$ time. It also immediately follows from Theorem~\ref{counting-SAT} that we can solve $\#$SAT on bounded fan-in circuits of depth $o(n)$ in $2^{n/2+o(n)}$ time, as such circuits can always be expressed as formulas of $\exp(o(n))$ size. It is also clear from the proof that we can trade off proof length and verification time: if we restrict the proofs to have length $2^\ell \leq 2^{n/2}$ (so that Merlin sends a polynomial of degree roughly $2^\ell$), then verifying the remaining sum over $n-\ell$ variables takes $O^{\star}(2^{n-\ell})$ time. 

We also observe that with more rounds of interaction, Merlin and Arthur can use shorter proofs. This is somewhat expected, because it is well-known that in $O(n)$ rounds, we can compute $\#$SAT with $\poly(n)$ communication and $\poly(n)$ verification time~\cite{Lund-Fortnow-Karloff-Nisan90}.

\begin{theorem} \label{counting-SAT2} For any $k > 0$, and $c > 2$, $\# SAT$ for Boolean formulas with $n$ variables and $m$ connectives has an interactive proof system with $c$ rounds of interaction, using $2^{n/(c+1)}\cdot \poly(n,m)$ time with randomness $O(n)$ and error probability $1/\exp(n)$.
\end{theorem}

\begin{proof} (Sketch) We essentially interpolate between our protocol and the LFKN protocol for $\# SAT$. Let $F$ be a Boolean formula over AND, OR, and NOT with $n$ variables and $m$ connectives, and let $P$ be its arithmetization as in Theorem~\ref{counting-SAT}. We will work modulo a prime $p > 2^n$, as before. For simplicity let us assume $n$ is divisible by $c+1$, and that $m \leq 2^{o(n)}$. Partition the set of variables into subsets $S_1,\ldots,S_{c+1}$ of $n/(c+1)$ variables each. Via interpolation, define the polynomials $\Psi_1,\ldots,\Psi_{\frac{n}{c+1}}$ analogously to Theorem~\ref{multipoint-eval}, where for all $j \in \{0,1,\ldots,2^{n/(c+1)}-1\}$, $\Psi_i(j)$ outputs the $i$th bit of the $j$ in $n/(c+1)$-bit binary representation. Now consider the polynomial in $c+1$ variables:\[Q_1(y) := \sum_{\substack{j_2,\ldots,j_{c+1} \\\in \{0,1,\ldots,2^{n/(c+1)}-1\}}} P(\Psi_1(y), \ldots,\Psi_{\frac{n}{c+1}}(y),\Psi_1(j_2),\ldots,\Psi_{\frac{n}{c+1}}(j_2),\ldots\ldots,\Psi_{\frac{n}{c+1}}(j_{c+1})).\] In the first round of interaction, an honest prover sends $Q_1(y)$, which has degree $2^{n/(c+1)+o(n)}$. The verifier then chooses a random $r_1 \in \F_p$, and sums $Q_1(y)$ over all points $\{0,1,\ldots,2^{n/(c+1)}-1\}$. 

In the $k$th round of interaction for $k=2,\ldots,c$, the honest prover sends the $2^{n/(c+1)+o(n)}$-degree polynomial \[Q_k(y) := \sum_{\substack{j_{k+1},\ldots,j_{c+1}\\ \in \{0,1,\ldots,2^{n/(c+1)}-1\}}} P(\Psi_1(r_1),\ldots,\Psi_{\frac{n}{c+1}}(r_{k-1}),\Psi_1(y),\ldots,\Psi_{\frac{n}{c+1}}(y),\Psi_{1}(j_{c+1})\ldots,\Psi_{\frac{n}{c+1}}(j_{c+1})).\] The verifier again chooses a random $r_k \in \F_p$. 

Finally in the $c$th round, after the prover has sendt $Q_c(y)$ and the verifier has chosen $r_c \in \F_p$ at random, the remaining computation is to compute the sum $\sum_{i=0}^{2^{n/(c+1)}-1} Q_c(j_i)$, and to verify that \[\sum_{\substack{j_{c+1} \\\in \{0,1,\ldots,2^{n/(c+1)}-1\}}} 
P(\Psi_1(r_1),\ldots,\Psi_{\frac{n}{c+1}}(r_1),\ldots,\Psi_1(r_c),\ldots,\Psi_{\frac{n}{c+1}}(r_c),\Psi_1(j_{c+1}),\ldots,\Psi_{\frac{n}{c+1}}(j_{c+1})) 
= Q_c(r_c).\] In each of the $c$ rounds, the chance of picking a ``bad'' $r_i$ is at most $2^{\frac{n}{c+1}+o(n)}/p \leq \exp(-\Omega(n))$.
\end{proof}

Thus, with $\omega(1)$ rounds of interaction, Arthur and Merlin can compute $\# SAT$ in $2^{o(n)}$ verification time and communication.

\subsection{Univariate Polynomial Identity Testing and the Nondeterministic SETH} \label{upit-nseth}
A nice aspect of Theorem~\ref{VP-sim} and its corollaries is that the randomness is low: for example, the obvious derandomization strategy of simulating all $O^{\star}(2^{n/2})$ coin tosses recovers a nondeterministic $O^{\star}(2^n)$ time algorithm for counting SAT assignments modulo $2$. 

The proof system itself motivates the following problem. Let \emph{univariate polynomial identity testing} (UPIT) be the problem of testing identity for two arithmetic circuits with \emph{one} variable, degree $n$, and $O(n)$ wires, over a field of order $\poly(n)$. The following corollary is immediate from the proofs of Theorems~\ref{main-multipoint},~\ref{counting-SAT}, and the above observations:

\begin{corollary} \label{UPIT-NSETH} If $\text{UPIT} \in \NTIME[n^{2-\eps}]$ for some $\eps > 0$, then $\#$Circuit-SAT for $o(n)$-depth circuits is computable in nondeterministic $2^{n(1-\eps/2)+o(n)}$ time. By~\cite{Williams11,JahanjouMV15,CarmosinoGIMPS15}, this further implies that $\E^{\NP}$ does not have $2^{o(n)}$-size sublinear-depth circuits.
\end{corollary}

In particular, the randomized verification task of Arthur in the protocol of Theorem~\ref{counting-SAT} directly reduces to solving UPIT on two univariate circuits of degree $2^{n/2+o(n)}$ and size $2^{n/2+o(n)}$. Hence, assuming the hypothesis of Corollary~\ref{UPIT-NSETH}, Arthur's verification can be performed deterministically in $2^{n(1-\eps/2)+o(n)}$ time.
 
This is an intriguing example of how derandomization \emph{within polynomial time} can imply strong circuit lower bounds: it is easy to see that UPIT is solvable in $\tilde{O}(n)$ time with randomness, and in $\tilde{O}(n^2)$ time deterministically, by efficient interpolation on $n+1$ distinct points (Theorem~\ref{interpolation}). In all other cases we are aware of (such as \cite{KI04,Williams10}), the necessary derandomization problem is only known to be solvable in deterministic \emph{exponential} time. Thus, the Nondeterministic SETH predicts that the exponent of the simple $\tilde{O}(n^2)$ algorithm for UPIT cannot be improved, even with nondeterminism.

\section{Quantified Boolean Formulas}

In the previous section, we saw how generic $\#P$ counting problems can be certified faster than exhaustive search. We can also give less-than-$2^n$ time three-round proof systems for certifying quantified Boolean formulas, a $\PSPACE$-complete problem. Our quantified Boolean formulas have the form \[(Q_1 x_1) \cdots (Q_n x_n) F(x_1,\ldots,x_n),\] where $F$ is an arbitrary formula on $m$ connectives, and each $Q_i \in \{\exists,\forall\}$.

\begin{reminder}{Theorem~\ref{QBF}}
Quantified Boolean Formulas with $n$ variables and $m \leq 2^n$ connectives have a three-round interactive proof system running in $2^{2n/3} \cdot \poly(n,m)$ time with $O(n)$ bits of randomness.
\end{reminder}

\begin{proof} Let $\phi = (Q_1 x_1) \cdots (Q_n x_n) F(x_1,\ldots,x_n)$ be a quantified Boolean formula to certify. Let $\delta > 0$ be a parameter to set later. First, convert the propositional formula $F'$ to an equivalent arithmetic circuit $P$ of $\poly(m)$ degree and size, as in Theorem~\ref{counting-SAT}. Note that $P$ outputs $0$ or $1$ on every Boolean input to its variables. Next, determine whether the quantifier suffix $(Q_{n-\delta n+1} x_{n-\delta n + 1}) \cdots (Q_n x_n)$ contains at least as many existential quantifiers as universal quantifiers. 

$\bullet$ If there are more existentially quantified variables, convert the subformula \[\phi'(x_1,\ldots,x_{n-\delta n}) = (Q_{n-\delta n+1} x_{n-\delta n+1}) \cdots (Q_n x_n)P(x_1,\ldots,x_n)\] into an arithmetic formula $P'$ in a standard way, where each $(\exists x_i)$ is replaced by a sum over $x_i \in \{0,1\}$, and each $(\forall x_i)$ is replaced by a product over $x_i \in \{0,1\}$. The formula $P'$ has size $2^{\delta n}\cdot \poly(m)$, for the tree of possible assignments to the last $2^{\delta n}$ variables times the size of the polynomial $P$. 

It is easy to see that $P'(a_1,\ldots,a_{n-\delta n})$ is nonzero (over $\Z$) on a Boolean assignment $(a_1,\ldots,a_{n-\delta n})$ if and only if $\phi'(a_1,\ldots,a_{n-\delta n})$ is true. Moreover, $P'$ has degree at most $\poly(m) \cdot 2^{\delta n/2}$, since there are most $\delta n/2$ universal quantifiers among the last $\delta n$ variables (so the $2^{\delta n}$ tree contains at most $\delta n/2$ layers of multiplication gates). Note the value $V_{a_1,\ldots,a_{n-\delta n}} = P'(a_1,\ldots,a_{n-\delta n})$ is always at most $(2^n \cdot m)^{O(2^{\delta n/2})}$. 

Our protocol begins by having Arthur send a random prime $p$ from the interval $[2,2^{2n^2}\cdot m]$ to Merlin, to help reduce the size of the values $V_{a_1,\ldots,a_{n-\delta n}}$. (A similar step also occurs in the proof that ${\sf IP} = \PSPACE$~\cite{Shamir90,Shen92}.) Since a nonzero $V_{a_1,\ldots,a_{n-\delta n}}$ has at most $O(2^{\delta n/2}nm)$ prime factors, the probability that a random $p \in [2,2^{2n^2}\cdot m]$ divides a fixed $V_{a_1,\ldots,a_{n-\delta n}}$ is at most \[\frac{O(2^{\delta n/2}n(n+\log m))}{2^{2n^2}},\] by the Prime Number Theorem. By the union bound, $p$ divides $V_{a_1,\ldots,a_{n-\delta n}}$ for some $a_1,\ldots,a_{n-\delta n} \in \{0,1\}$ with probability at most $(\log m)/2^{\Omega(n^2)}$. Therefore for all $a_1,\ldots,a_{n-\delta n} \in \{0,1\}$, the ``non-zeroness'' of $P'(a_1,\ldots,a_{n-\delta n})$ over $\Z$ is preserved over the field $\F_p$, with high probability. Merlin and Arthur will work over $\F_p$ in the following.

Applying Theorem~\ref{multipoint-eval} to $P'$ with $d := \poly(m) \cdot 2^{\delta n/2}$, $p := 2^{2n}\cdot m$, $K := 2^{n-\delta n}$, and $s := \poly(m) \cdot 2^{\delta n}$, there is an MA-proof system where Merlin sends a proof of length at most $2^{n-\delta n/2} \cdot \poly(n)$ bits, while Arthur uses at most $\poly(n)$ coins and $(2^{n-\delta n/2} + 2^{\delta n})\cdot \poly(n,m)$ time, outputting the value of $P'$ on all $2^{n-\delta n}$ Boolean inputs with high probability. It is easy to determine the truth value of the original QBF $\phi$ from the $2^{n-\delta n}$-length truth table of $P'$; this is simply a formula evaluation on an $O(2^{n-\delta n})$-size formula defined by the quantifier prefix $(Q_1 x_{1}) \cdots (Q_{n-\delta n} x_{n-\delta n})$.

Setting $\delta = 2/3$ yields a $2^{2n/3}\cdot \poly(n,m)$-length proof and an analogous running time bound. 

$\bullet$ If there are at least as many universal variables as existential ones, then Merlin and Arthur decide to prove that $\neg \phi$ is false, by flipping the type of every quantifier (from existential to universal, and vice-versa) and replacing $P$ with an arithmetic circuit for $\neg F$. Now the quantifier suffix $(Q'_{n-\delta n+1} x_{n-\delta n + 1}) \cdots (Q'_n x_n)$ of the new QBF contains more existential quantifiers than universal ones, and we proceed as in the first case, evaluating an $(n-\delta n)$-variable formula of $2^{\delta n}$ size (and at most $\delta n/2$ universally quantified variables) on all of its possible assignments, and inferring the truth or falsity of the QBF from that evaluation.\end{proof} 

\section{Conclusion}

By a simple but powerful protocol for batch multipoint evaluation, we have seen how non-interactive proof systems can be exponentially more powerful than randomized or nondeterministic algorithms, assuming some exponential-time hypotheses. There are many questions left to pursue, for instance:

\begin{compactitem}
\item Are there more efficient proof systems if we just want to prove that a formula is UNSAT? Perhaps UNSAT has an MA-proof system of $O^{\star}(2^{n/3})$ time. Perhaps Parity-SAT could be certified more efficiently, exploiting the nice properties of characteristic-two fields? By the Valiant-Vazirani lemma~\cite{Valiant-Vazirani86}, this would imply a three-round interactive proof system for UNSAT that is also more efficient. Our MA-proof systems all have extremely low randomness requirements of Arthur. If we allowed $2^{\delta n}$ bits of randomness for some $\delta > 0$, perhaps they can be improved further. 

\item Faster nondeterministic UNSAT algorithms are by now well-known to imply circuit lower bounds for problems in nondeterministic exponential time~\cite{Williams10,JMV13}. Can the proof systems of this paper be applied to conclude new lower bounds? One difficulty is that we already know ${\sf MAEXP} \not\subset \P/\poly$~\cite{Buhrman-Fortnow-Thierauf98}. More seriously, it seems possible that one could apply our protocol for $\#$SAT on circuits of $o(n)$ depth to show that (for instance) $\E^{\text{\sf promise}\MA}$ does not have $2^{o(n)}$ size formulas; this would be a major advance in our understanding of exponential-size circuits. 

\item Can $O^{\star}(2^{n/2})$-time Merlin-Arthur proof system for $\#SAT$ be converted into a construction of nondeterministic circuits of $(2-\eps)^n$ size for UNSAT? To do this, we would want to have a small collection of coin tosses that suffices for verification. If we convert the proof system into an Arthur-Merlin game in the standard way, the protocol has the following structure: for a proof-length parameter $\ell$, we can toss $O(\ell \cdot n)$ random coins are tossed prior to the proof, then Merlin can give a single $\tilde{O}(\ell)$-bit proof of the protocol that needs to be simulated on $O(\ell)$ different coin tosses of $n/2+\tilde{O}(1)$ bits each. The difficulty is that each of these $O(\ell)$ coin tosses takes $\Omega(2^n/\ell)$ time for Arthur to verify on his own, as far as we can tell. So even though the probability of error here could be extremely small (less than $1/2^{\Omega(\ell)}$) we do not know how to get a $(2-\eps)^n$ time algorithm for verification.

\item Does QBF on $n$ variables and $\poly(n)$ connectives have an MA-proof system using $(2-\eps)^n$ time, for some $\eps > 0$?
\end{compactitem}

\paragraph{Acknowledgements.} I thank Russell Impagliazzo for sending a draft of his paper (with coauthors) on NSETH, MASETH, and AMSETH, and for discussions on the $\#$SAT protocol. I also thank Petteri Kaski for suggesting that I add a protocol for Closest Pair and Hamiltonian Cycle, and Shafi Goldwasser for references.

\bibliographystyle{alpha}
\bibliography{cc-papers}

\appendix

\section{Quick Proof Systems For Some Poly-Time Problems}\label{ov-section}

We can also obtain nearly-linear time MA-proof systems for quite a few problems which have been conjectured to be hard to solve faster than quadratic time. Perhaps the most illustrative example is a proof system for computing orthogonal pairs of vectors. Via reductions, this result implies analogous proof systems for several other quadratic-time solvable problems (see~\cite{AbboudWY15}); we omit the details here.

\begin{theorem} \label{ov2} Let $d \leq n$. For every $A \subseteq \{0,1\}^d$ such that $|A|=n$, there is a MA-proof system certifying for every $v \in A$ if there is a $u \in A$ such that $\langle v,u\rangle = 0$, with $\tilde{O}(n\cdot d)$ time and error probability $1/\poly(n)$.
\end{theorem}

\begin{proof} Let $p$ be a prime greater than $n^2 \cdot d$. 
Define the $2d$-variable polynomial \[P(x_{1},\ldots,x_{d},y_1,\ldots,y_d) := \prod_{i=1}^d  \left(1-x_{i}\cdot y_{i}\right).\] Observe $\deg(P) \leq 2d$, and for a pair of Boolean vectors $u,v \in \{0,1\}^d$, $P(u,v) = 1$ if $\langle u,v\rangle = 0$, otherwise $P(u,v) = 0$. Then, the polynomial \begin{align*} P'(u[1],\ldots,u[d]) := \sum_{j=1,\ldots,n} P(u[1],\ldots,u[d],v_j[1],\ldots,v_j[d])\end{align*} counts the number of vectors in $A$ that are orthogonal to the input vector $u \in \{0,1\}^d$. Note the size of $P'$ as an arithmetic circuit is $O(n \cdot d)$, and its degree is at most $2d$ as well. Applying Theorem~\ref{multipoint-eval} directly, we can certify the evaluation of $P'$ on all $n$ vectors of $d$ dimensions in $\tilde{O}(n\cdot d)$ time.  
\end{proof}

One consequence (among many) of Theorem~\ref{ov2} is an MA-proof system for the \emph{dominating pairs} problem in computational geometry: given a set $S$ of $n$ vectors in $\R^d$, determine if there are $u,v\in S$ such that $u[i] < v[i]$ for all $i=1,\ldots,d$. (Here, our computational model is the real RAM, where additions and comparisons of reals are unit time operations.)
\begin{corollary} There is an MA-proof system for counting the number of dominating pairs in $\tilde{O}(n^{1.5}\cdot d^{1.5})$ time. As a consequence, there is a MA-proof system for counting $0$-$1$ solutions to a linear program with $k$ variables and $m$ constraints that runs in $2^{3k/4} \cdot \poly(m,k)$ time.
\end{corollary}

\begin{proof} Given that one can count orthogonal vectors of $n$ vectors in $d$ Boolean dimensions in $t(n,d)$ time, a recent reduction of Chan and the author~\cite{Chan-Williams16} shows how to count the number of dominating pairs among $n$ vectors in $\R^d$, in $O(n^2 d^2/s + t(n,2+ds))$ time, for any positive natural number $s$. In fact, the reduction makes precisely one call to orthogonal vectors. Theorem~\ref{ov2} provides an $\tilde{O}(n\cdot d)$ time proof system for counting orthogonal vectors, so by setting $s = \sqrt{n\cdot d}$ to balance the factors, there  is a proof system for counting dominating pairs in $\tilde{O}(n^{1.5}\cdot d^{1.5})$ time. By a reduction of Impagliazzo, Paturi, and Schneider~\cite{IPS13} from integer linear programming to dominating pairs, we obtain an MA-proof system for counting the number of Boolean solutions to a linear program with $k$ variables and $m$ inequalities in $2^{3k/4} \cdot \poly(m,k)$ time.
\end{proof}

Finally, we illustrate that the above ideas can certify Nearest Neighbors (in the Hamming metric) in near-linear time as well:

\begin{reminder}{Theorem~\ref{hamming}} Let $d \leq n$. For every $A \subseteq \{0,1\}^d$ with $|A|=n$, and every parameter $k \in \{0,1,\ldots,d\}$, there is an MA-proof system certifying for every $v \in A$ the number of points in $A$ with Hamming distance at most $k$ from $v$, running in $\tilde{O}(n\cdot d)$ time with error probability $1/\poly(n)$.
\end{reminder}

\begin{proof} (Sketch) Analogous to Theorem~\ref{ov2}.
Let $p$ be a prime greater than $n^2 \cdot (2d+1)$, and let $k \in \{0,1,\ldots,d\}$ be our proximity parameter. Define the degree-$2d$ polynomial $\Psi(x)$ to be $0$ on all $j=-d,\ldots,d-2k$, and $1$ on all $j=d-2k,\ldots,d$. Note that such a $\Psi$ can easily be constructed by interpolation in $\tilde{O}(d)$ time (cf. Theorem~\ref{interpolation}). Define the $2d$-variable polynomial \[P(x_{1},\ldots,x_{d},y_1,\ldots,y_d) := \Psi\left(\sum_{i=1}^d x_{i}\cdot y_{i}\right).\] Observe that $\deg(P) \leq 2d$, and for a pair of Boolean vectors $u,v \in \{-1,1\}^d$, $P(u,v) = 1$ if and only if $u$ and $v$ differ in at most $k$ coordinates. (Differing in $k$ coordinates is equivalent to summing $(d-k)$ ones and $k$ minus-ones in the inner product.) Therefore, if we map all the $0/1$ vectors in $A$ to $1/-1$ vectors (mapping $0$ to $1$, and mapping $1$ to $-1$), the polynomial \begin{align*} P'(u[1],\ldots,u[d]) := \sum_{j=1,\ldots,n} P(u[1],\ldots,u[d],v_j[1],\ldots,v_j[d])\end{align*} counts the number of vectors in $A$ (construed as vectors in $\{-1,1\}$, instead of $\{0,1\}$) that have Hamming distance at most $k$ from the input $u \in \{-1,1\}^d$. The size of $P'$ is $O(n \cdot d)$, its degree is at most $2d$, and applying Theorem~\ref{multipoint-eval} allows us to certify the evaluation of $P'$ on all $n$ vectors of $d$ dimensions in $\tilde{O}(n\cdot d)$ time. Our prime $p$ is chosen large enough so that the values of all intermediate computations are preserved.   
\end{proof}

\subsection{Certifying the Number of Small Cliques}

The final result of this section gives an efficient MA-proof system for verifying the number of $k$-cliques in a graph:

\begin{reminder}{Theorem~\ref{k-clique}} For every $k$, there is a MA-proof system such that for every graph $G$ on $n$ nodes, the verifier certifies the number of $k$-cliques in $G$ using $\tilde{O}(n^{\lfloor k/2\rfloor+2})$ time, with error probability $1/\poly(n)$.
\end{reminder}

\begin{proof} The strategy (as in previous proofs) is to reduce the problem to multipoint evaluation of an appropriate circuit on an appropriate list of points, and appeal to Theorem~\ref{multipoint-eval}. 

Given a graph $G=(V,E)$ on $n$ nodes with $V = [n]$, let $A$ be its adjacency matrix. Let $\ell\text{-Cliques}(G)$ be the collection of all $\ell$-cliques of $G$, represented as subsets of $[n]$ of cardinality $\ell$. Given a subset $S \subseteq [n]$, let $J(S) := \{v \in (V-S)~|~(\forall u \in S)[(u,v) \in E]\}$ be the \emph{joint neighborhood of $S$}. We denote the members of $J(S)$ as $\{u_{J(S),1},\ldots,u_{J(S),|J(S)|}\} \subseteq [n]$. Consider the polynomial
\[C(x_1,\ldots,x_n) := \sum_{S \in \ell\text{-Cliques}(G)}  E^{k-\ell}_{|J(S)|}(x_{u_{J(S),1}},\ldots,x_{u_{J(S),|J(S)|}}),\]
where $E^k_n$ is the $k$th elementary symmetric polynomial on $n$ variables. Suppose $a = (a_1,\ldots,a_n) \in \{0,1\}^n$ contains exactly $k-\ell$ ones, and let $T_a \subseteq [n]$ be the set corresponding to $a$. Observe that $C(a_1,\ldots,a_n)$ equals the number of $S \subseteq (V-T_a)$ such that $S$ is an $\ell$-clique and every node of $S$ has an edge to every node of $T_a$. Therefore, if we evaluate $C$ on the indicator vectors for every $(k-\ell)$-clique in $G$, the sum of these evaluations will be the number of $k$-cliques in $G$ times ${n \choose k-\ell}$ (every $k$-clique will be counted ${n \choose k-\ell}$ times in the summation). 

Therefore, it suffices to evaluate $C$ on the $O({n \choose k-\ell})$ indicator vectors of $(k-\ell)$-cliques in $G$. These vectors of length $n$ can obviously be prepared in $O(n^{k-\ell+1})$ time. 

It is well-known that for every $k$, the $k$th elementary symmetric polynomial on variables $x_1,\ldots,x_n$ can be computed in $O(n^2)$ size and degree $O(n)$ (this result is often attributed to Ben-Or). To compute this polynomial, we just have to determine the coefficient of $z^k$ in the polynomial \[\prod_{i=1}^n(z-x_i),\] which can be done by computing the coefficient of $z^k$ in the polynomial determined by feeding the set of points $\{(x_0,\prod_{i=1}^n(x_0-x_i)),(x_1,0),\ldots,(x_n,0)\}$ into a circuit for univariate interpolation, where $x_0$ is a point different from $x_1,\ldots,x_n$. Each of the joint neighborhoods $J(S)$ can easily be determined in $O(\ell \cdot n)$ time. The total degree of $C$ is therefore $O(n)$, and its size is $O(n^2 \cdot {n \choose \ell})$.

Applying Theorem~\ref{multipoint-eval} directly, we can evaluate $C$ on $O({n \choose k-\ell})$ points over $\F_p$ with $p > n^{k-\ell}$, in time \[\tilde{O}\left({n \choose k-\ell}\cdot n + {n \choose \ell}\cdot n^2\right).\] Setting $\ell=\lfloor k/2 \rfloor$ yields a running time of $\tilde{O}(n^{\lfloor k/2\rfloor+2})$.
\end{proof}

\end{document}